\documentclass{elsarticle}
\usepackage{array,xspace,multirow,hhline,tikz,colortbl,tabularx,booktabs,fixltx2e,amsmath,amssymb,amsfonts,amsthm,multirow}
\usepackage{algorithm}
\usepackage{graphics}
\usepackage{algorithmic}
\usepackage{verbatim,ifthen}
\usepackage{enumitem}
\usepackage{pifont}
\usepackage{ifthen}
\usepackage{nicefrac}
\usepackage{calrsfs,mathrsfs}
\usepackage{bbding,pifont}
\usepackage{pgflibraryshapes}
\usetikzlibrary{trees}
\usetikzlibrary{positioning,chains,fit,shapes,calc}
\usepackage{tkz-graph}

\usepackage{eucal}

\usepackage{tikz}
\usetikzlibrary{arrows}
\usetikzlibrary{decorations.pathreplacing}

	\usepackage{varioref}

\definecolor{light-gray}{gray}{0.9}

	\newcommand{\eg}{e.g.,\xspace}

	\newtheorem{lemma}{Lemma}%
		\newtheorem{remark}{Remark}%
	\newtheorem{theorem}{Theorem}%
	\newtheorem{corollary}{Corollary}%
	\newtheorem{example}{Example}
		\newtheorem{claim}{Claim}
			\newtheorem{fact}{Fact}

	\newcommand\eat[1]{}

	\usepackage{enumitem}
	\setenumerate[1]{label=\rm(\it{\roman{*}}\rm),ref=({\it\roman{*}}),leftmargin=*}
	\newlength{\wordlength}

	\newcommand{\set}[1]{\{#1\}}
	\newcommand{\midd}{\mathbin{:}}

	\newcommand{\eqclass}[2][]{\ifthenelse{\equal{#1}{}}{[#2]}{[#2]_{\sim_{#1}}}}

\RequirePackage{ifthen,calc}
\newsavebox{\ffbox}\newlength{\ffboxlen}
\newcommand{\todo}[1]{%
	\vskip4mm
	{\sbox{\ffbox}{\textbf{\color{red} TODO:}\ \textit{{#1}}\ \textbf{\color{red} :ODOT}}
    \settowidth{\ffboxlen}{\usebox{\ffbox}}
		\addtolength{\ffboxlen}{-5mm}
    \ifthenelse{\ffboxlen>\linewidth}{%
	\noindent\marginpar{$>>>>$}\textbf{\color{red} TODO:}\ \textit{{#1}}\ \textbf{\color{red} :ODOT}\marginpar{$<<<<$}}{%
	  \noindent\marginpar{$>><<$}\textbf{\color{red} TODO:}\ \textit{{#1}}\ \textbf{\color{red} :ODOT}}}
	\vskip4mm
  }

	\newcommand{\pref}{\succsim \xspace}
	\newcommand{\Pref}[1][]{
		\ifthenelse{\equal{#1}{}}{\mathrel R}{\mathop{R_{#1}}}
	}                           

	\newcommand{\spref}{\ensuremath{\succ}}
	               
	\newcommand{\sPref}[1][]{                  
		\ifthenelse{\equal{#1}{}}{\mathrel P}{\mathop{P_{#1}}}
	}                                          
	\newcommand{\Indiff}[1][]{                 
		\ifthenelse{\equal{#1}{}}{\mathrel \sim}{\mathop{I_{#1}}}
	}
	\newcommand{\prefset}[1][]{\ifthenelse{\equal{#1}{}}{\mathcal{R}}{\mathcal{R}_{#1}}}

\usepackage{enumitem}
\setenumerate[1]{label=\rm(\it{\roman{*}}\rm),ref=({\it\roman{*}}),leftmargin=*}

\newcommand{\nbh}[1][]{
	\ifthenelse{\equal{#1}{}}{\nu}{\nu(#1)}
}

\newcommand{\cstr}[1][]{
	\ifthenelse{\equal{#1}{}}{\mathscr S}{\cstr(#1)}
}

\newcommand{\choice}[1][]{
	\ifthenelse{\equal{#1}{}}{\mathit{C}}{\choice(#1)}
}

\usepackage{boxedminipage}
\usepackage{xspace}

\tikzset{
  treenode/.style = {align=center, inner sep=0pt, text centered,
    font=\sffamily},
  arn_n/.style = {treenode, circle, white, font=\sffamily\bfseries, draw=black,
    fill=black, text width=1.5em},
  arn_r/.style = {treenode, circle, red, draw=red, 
    text width=1.5em, very thick},
  arn_x/.style = {treenode, rectangle, draw=black,
    minimum width=0.5em, minimum height=0.5em}
}

\newcommand{\mc}{\mathcal C}

\newcommand{\ra}{\rightarrow}

\newcommand{\Omit}[1]{}

\begin{document}

\title{Structure and complexity of \\ex post efficient random assignments}
	

 \author{Haris Aziz and Simon Mackenzie
 		}
 	 \address{NICTA and UNSW, Australia}
 	\author{Lirong Xia
 			}
		\address{Rensselaer Polytechnic Institute (RPI), USA}
	 
 	\author{Chun Ye
 			}
		\address{Columbia University, USA}
	
	


\begin{abstract}
	In the random assignment problem, objects are randomly assigned to agents keeping in view the agents' preferences over objects. 
	A random assignment specifies the probability of an agent getting an object. We examine the structural and computational aspects of ex post efficiency of random assignments.  We first show that whereas an ex post efficient assignment can be computed easily, checking whether a given random assignment is ex post efficient is NP-complete.
  		Hence implementing a given random assignment via deterministic Pareto optimal assignments is NP-hard.
	We then formalize another concept of efficiency called robust ex post efficiency that is weaker than stochastic dominance efficiency but stronger than ex post efficiency. We present a characterization of robust ex post efficiency  and show that it can be tested in polynomial time if there are a constant number of agent types.
	It is shown that the well-known random serial dictatorship rule is not robust ex post efficient. 
	Finally, we show that whereas robust ex post efficiency depends solely on which entries of the assignment matrix are zero/non-zero, ex post efficiency of an assignment depends on the actual values.

\end{abstract}

	\begin{keyword}
	 	Random assignment
		\sep efficiency\\
		
		\emph{JEL}: C62, C63, and C78
	\end{keyword}

\maketitle


Pareto optimality has been termed the \emph{``single most important tool of normative economic analysis''}~\citep{Moul03a}. 
It appeals to the idea that there should not exist another possible outcome different from the social outcome which all the agents prefer. 
We consider Pareto optimality in the \emph{random assignment problem} which is a fundamental and widely applicable setting in computer science and economics~\citep[see \eg][]{Sven99a, BoMo01a, KaSe06a,BCK11a,BCKM12a,Aziz14d}.

In a \emph{random assignment problem} $(N,O,\pref)$, there is a set of agents $N=\{1,\ldots, n\}$, a set of objects $O=\{o_1,\ldots, o_n\}$, and a preference profile $\pref=(\pref_1,\ldots, \pref_n)$ that specifies for each agent $i\in N$ his \textit{strict} preferences over objects in $O$. The goal is to find a desirable assignment keeping in view the preferences of the agents.
A \emph{random assignment} which we will simply refer to as assignment assigns the probability of agents getting objects.
A random assignment can be represented as a bistochastic matrix in which each entry denotes the probability of an agent getting an object. Since both the probability of an agent getting some object and the probability that an object is allocated to some agent is one, each column and row of the random assignment matrix sums up to one.  
A \emph{deterministic assignment} is a random assignment in which the probability of an agent getting an object is either one or zero. The advantage of using random assignments instead of deterministic assignments is that they can allow for better ex ante fairness. It is well-known that any random assignment can be a result of a probability distribution over deterministic assignments~\citep{Birk46a}. A random assignment also a useful time-sharing interpretation whereby the probability of agent $i$ getting object $o$ is the fraction of time he will be matched to object $o$~\citep[see e.g., ][]{RRV93a,BCKM12a}.

In this paper, we focus on efficiency of random assignments. 
A deterministic assignment $p$ is \emph{Pareto optimal} if there exists no other deterministic assignment $q$ such that each agent weakly prefers his object allocated in assignment $q$ and at least one agent strictly prefers his object allocated in assignment $q$. When the assignment is random, Pareto optimality can be generalized to two well-studied efficiency concepts --- ex post efficiency  and stochastic dominance (SD) efficiency. A random assignment is \emph{ex post efficient} if it can be represented as a convex combination of Pareto optimal deterministic assignments. A random assignment is SD-efficient is there exists no other random assignment which each agent weakly  prefers and some agent strictly prefers with respect to the stochastic dominance relation.
Ex post efficiency is a weaker requirement than  \emph{stochastic dominance (SD) efficiency}~\citep{KaSe06a}.

The main research problem in this paper is to understand \emph{the structure and complexity of efficient assignments in particular ex post efficiency assignments}. We not only consider ex post efficiency and SD-efficiency but also introduce an intermediate notion called \emph{robust ex post efficiency} that is weaker than SD-efficiency and stronger than ex post efficiency.  
We seek to understand the geometry of the ex post efficient polytope and where the robust ex post efficient and SD-efficient points lie within the ex post efficient polytope or the assignment polytope.  An efficiency concept is deemed \emph{combinatorial} if the efficiency of an assignment solely depends on which entries of the assignment matrix are zero or non-zero. We explore which of the efficiency concepts are combinatorial. We also consider natural computational problems related to efficiency of random assignments. Previously, computational aspects of Pareto optimal deterministic assignments have been studied in great depth in recent years~\citep{ACMM05a,ABH11c,Manl13a,ABM14b}. Similar analysis has been done for SD-efficient assignments where it has been shown that
not only can an SD-efficient random assignment be computed efficiently~\citep{BoMo01a}, a linear programming formulation can be used to check whether an assignment is SD-efficient or not~\citep{Atha11a}. However, to the best of our knowledge, the complexity of testing ex post efficiency has not been settled. Testing ex post efficiency is also closely related to implemeting a random assignment with respect to discrete Pareto optimal assignments. 

If one is able to compute an SD-efficient assignment~\citep{BoMo01a}, then the question arises that why should we bother with a less demanding notion of efficiency? There are a number of reasons why implementation of ex post assignments and testing ex post efficiency is important.
Firstly, the algorithm to test SD-efficiency of a random assignment cannot be used to test weaker notions of efficiency.
Secondly, in many scenarios, a random assignment may be given \emph{a priori} because of various constraints and hence may not be SD-efficient. For example, the random assignment could be a result of an already decided time sharing agreement. Such a random assignment may need to be implemented in any case and would preferably  be implemented via Pareto optimal deterministic assignments.
For example, agents may already have a time sharing assignment in place and one may want to know whether it can be achieved by randomizing over deterministic Pareto optimal assignments. Thirdly, there may already be simple strategyproof method such as the uniform assignment rule in place where each agent gets $1/n$ of each object.\footnote{The uniform assignment also satisfies other properties such a probabilistic consistency~\citep{Cham04a}} One may want to implement the uniform assignment via a convex combination of Pareto optimal assignments even if it may not be SD-efficient. We also note that SD-efficiency is incompatible with strategyproofness when also requiring anonymity~\citep{BoMo01a}. Finally, the convex hull of deterministic Pareto optimal assignments is an interesting mathematical object and testing ex post efficiency of a random assignment is equivalent to checking whether a given assignment is in the convex hull. The problem has important connections with optimizing linear functions over this convex hull.

\paragraph{Contributions}

	 We first examine the problem of checking whether a given random assignment is ex post efficient and obtain insights into why the problem may be computationally challenging.  
 We show that whereas computing an ex post efficient assignment is easy, checking whether a given random assignment is ex post efficient is NP-complete.
Hence implementing a given random assignment via deterministic Pareto optimal assignments is NP-hard. Even if it is known that a random assignment is ex post efficient, finding its Pareto optimal decomposition is NP-hard. Our result also implies that optimizing over the convex hull of Pareto optimal assignments is NP-complete.
	
	 We formalize a new efficiency concept called \emph{robust ex post efficiency} that is weaker than SD-efficiency but stronger than ex post efficiency.  
	A characterization of robust ex post efficiency is also presented. Previously, characterizing SD-efficiency has already attracted considerable interest~\citep[see \eg][]{AbSo03a,Atha11a,BoMo01a}. 
	We show that robust efficiency can be checked in polynomial time if there are a constant number of agent types.	

	
	 We highlight that the well-known random serial dictatorship mechanism~\citep{ABB13b} is not robust ex post efficient. Our finding strengthens the observation of \citet{BoMo01a} that random serial dictatorship is not SD-efficient.

We show that whereas robust ex post efficiency is combinatorial, ex post efficiency is not. The finding that ex post efficiency is not combinatorial also contrasts with the fact that in randomized voting, ex post efficiency of a lottery simply depends on its support.

			Table~\ref{table:summary-expost} summarizes some of the results.  
					\begin{table*}[h!]
							\small
							\centering
								\scalebox{0.8}{
						\centering
						\begin{tabular}{lccc}
						\toprule
					&Ex post efficiency & Robust ex post efficiency& SD-efficiency \\
					\midrule
					\multirow{2}{*}{Complexity of verification}&NP-complete&in coNP, in P for const \# agent types&in P\\
					&(Theorem~\ref{mainthm})&(Remark~\ref{remark:incoNP}), (Lemma~\ref{lemma:agent-types})&(Theorem 1, \citep{Atha11a})\\
					\midrule
					\multirow{2}{*}{Combinatorial}&no&yes&yes\\
					&(Theorem~\ref{th:expost-not-combinatorial})&(Theorem~\ref{th:robust-is-combinatorial})&(Lemma 3, \citep{BoMo01a})\\
							\bottomrule
						\end{tabular}
						}

							\caption{Summary of results and related work.}
						\label{table:summary-expost}
						\end{table*}

\section{Preliminaries}

\paragraph{Assignment setting}

An assignment problem is a triple $(N,O,\pref)$ such that $N=\{1,\ldots, n\}$ is the set of agents, $O=\{o_1,\ldots, o_n\}$ is the set of objects, and the preference profile $\pref=(\pref_1,\ldots, \pref_n)$ specifies for each agent $i$ his preferences $\pref_i$ over objects in $O$. 
We write~$a \pref_i b$ to denote that agent~$i$ values object~$a$ at least as much as object~$b$ and use~$\spref_i$ for the strict part of~$\pref_i$, i.e.,~$a \spref_i b$ iff~$a \pref_i b$ but not~$b \pref_i a$. 
We will assume that the agents have strict preferences and that
$o_1\succ_i o_2 \cdots, o_n$ is represented by a comma separated list as follows:
\begin{align*}
	i: o_1,o_2,\ldots, o_n
\end{align*}

\begin{example}[Assignment Problem]
	Consider an assignment problem in which $N=\{1,2,3,4\}$, $O=\{o_1,o_2,o_3,o_4\}$ and the preferences $\pref$ are as follows.
		\begin{align*}
	1:&\quad o_1,o_2,o_3,o_4\\
	2:&\quad o_1,o_2,o_3,o_4\\
	3:&\quad o_2,o_1,o_4,o_3\\
	4:&\quad o_2,o_1,o_4,o_3
	\end{align*}
\end{example}


%
A random assignment $p$ is a $(n\times n)$ matrix $[p(i)(o_j)]$ such that $ p(i)(o_j) \in [0,1]$ for all $i\in N$, and $o_j\in O$, 
; $\sum_{i\in N}p(i)(o_j)= 1$ for all $o_j\in O$; and 
 $\sum_{o_j\in O}p(i)(o_j)= 1$ for all $i\in N$.
The value $p(i)(o_j)$ represents the probability of object $o_j$ being allocated to  agent $i$. Each row $p(i)=(p(i)(o_1),\ldots, p(i)(o_n))$ represents the allocation of agent $i$. 
The set of columns correspond to the objects $o_1,\ldots, o_n$.
A feasible random assignment is \emph{deterministic} if $p(i)(o)\in \{0,1\}$ for all $i\in N$ and $o\in O$.
A \emph{uniform assignment} is a random assignment in which each agent has probability $1/n$-th of getting each object. 

\begin{example}[Random assignment]
For and assignment problem in which $N=\{1,2,3,4\}$, $O=\{o_1,o_2,o_3,o_4\}$, the following is an example of a random assigment:
\[p=\begin{pmatrix}
	\nicefrac{5}{12}&\nicefrac{1}{12}&\nicefrac{5}{12}&\nicefrac{1}{12}\\
\nicefrac{5}{12}&\nicefrac{1}{12}&\nicefrac{5}{12}&\nicefrac{1}{12}\\
\nicefrac{1}{12}&\nicefrac{5}{12}&\nicefrac{1}{12}&\nicefrac{5}{12}\\
\nicefrac{1}{12}&\nicefrac{5}{12}&\nicefrac{1}{12}&\nicefrac{5}{12}
	\end{pmatrix}.\]
In $p$, the probability of agent $1$ getting $o_3$ is $5/12$.
\end{example}

Given two random assignments $p$ and $q$, $p(i) \succsim_i^{SD} q(i)$ i.e.,  an agent $i$ \emph{SD~prefers} allocation $p(i)$ to allocation $q(i)$ if 
	\[	\sum_{o_j\in \set{o_k\midd o_k\succsim_i o}}p(i)(o_j)\quad \ge \sum_{o_j\in \set{o_k\midd o_k\succsim_i o}}q{(i)(o_j)} \quad \text{ for all } o\in O.
	\]
	
		An assignment $p$ is \emph{SD-efficient} is there exists no assignment $q$ such that $q(i) \succsim_i^{SD} p(i)$ for all $i\in N$ and $q(i) \succ_i^{SD} p(i)$ for some $i\in N$. An assignment is \emph{ex post efficient} if it be can represented as a probability distribution over the set of Pareto optimal assignments.

We say that a deterministic assignment $q$ is \emph{consistent} with a random assignment $p$ if for each $q(i)(o)=1$, we have that  $q(i)(o)>0$.		
A deterministic assignment can be represented by a permutation matrix in which an entry of one denotes the row agent getting the column object.
A \emph{decomposition} of a random assignment $p$ is a sum $\sum_{i=1}^k \lambda_iP_i$ such that $\lambda_i\in (0,1]$ for $i\in \{1,\ldots,k\}$, $\sum_{i=1}^k\lambda_i=1$, and each $P_i$ is a permutation matrix (consistent with $p$). 

\begin{example}[Decomposition of a random assignment]
Consider a random assignment
\[\begin{pmatrix}
\nicefrac{1}{3}&\nicefrac{1}{3}&\nicefrac{1}{3}&0\\
\nicefrac{1}{3}&\nicefrac{1}{3}&0&\nicefrac{1}{3}\\
	\nicefrac{1}{3}&0&\nicefrac{2}{3}&0\\
0&\nicefrac{1}{3}&0&\nicefrac{2}{3}
	\end{pmatrix}\]

Then, the following is a valid decomposition of the assignment.	
	\[
\frac{1}{3}\begin{pmatrix}
1&0&0&0\\
0&1&0&0\\
	0&0&1&0\\
0&0&0&1
	\end{pmatrix}+\frac{1}{3}\begin{pmatrix}
0&0&1&0\\
0&0&0&1\\
	1&0&0&0\\
0&1&0&0
	\end{pmatrix}+\frac{1}{3}\begin{pmatrix}
0&1&0&0\\
1&0&0&0\\
	0&0&1&0\\
0&0&0&1
	\end{pmatrix}.\]
\end{example}

\begin{fact}
A deterministic assignment is Pareto optimal iff it is ex post efficient iff it is SD-efficient. 
\end{fact}

Therefore both SD-efficiency and ex post efficiency are natural generalizations of Pareto optimality in the context of random assignments. The convex hull of Pareto optimal discrete assignments will be denoted by $\mathcal{P}$.

An efficiency concept $X$ is \emph{combinatorial} if for any two random assignments $p$ and $q$ such that $q(i)(o)>0$ if and only if $p(i)(o)>0$, it holds that $p$ is efficient with respect to $X$ if and only if $q$ is efficient with respect to $X$.

%


%

\paragraph{Insights into ex post efficiency}

Before we examine ex post efficient random assignment, we review some characterizations of deterministic Pareto optimal assignments. The first characterization is with respect to deterministic assignment algorithm called \emph{serial dictatorship} which takes as a parameter a permutation $\pi$ over the set of agents.
Serial dictatorship $Prio(N,O,\pref,\pi)$ lets the agents in the permutation $\pi$ serially take their most preferred object that has not yet been allocated until each agent has an object.

\begin{fact}[\citet{AbSo98a}]\label{fact:po:sd}
Each Pareto optimal assignment is an outcome of applying serial dictatorship with respect to some permutation of the agents.
\end{fact}

It follows that RSD (random serial dictator) rule which takes some $\pi$ uniformly at random and then implements serial dictatorship with respect to it is ex post efficient.

%

The next characterization of Pareto optimal is graph-theoretical.
For an assignment problem $(N,O,\pref)$ and deterministic assignment $p$, the \emph{corresponding graph} $(V,E)$ is such that $V=N\cup O$ and $E$ is defined as follows. For all $i\in N$ and all $o\in O$,

\begin{itemize}
	\item $(o,i)\in E$ iff $p(i)(o)=1$ and
	\item $(i,o)\in E$ iff $o\succ_i o^*$ where $p(i)(o^*)=1$. 
\end{itemize}

		\begin{fact}\label{fact:po:ttc}
			An assignment is Pareto optimal if and only if its corresponding graph does not admit a cycle.
		\end{fact}
		
		The cycle is referred to as a trading cycle because it represents a trade of objects in which each agent in the cycle gets the object that he was pointing to~\cite[see \eg][]{AzKe12a}.
We will use both Facts \ref{fact:po:sd} and \ref{fact:po:ttc} in our arguments.


\begin{remark}
	An ex post efficient assignment can be computed in polynomial time. An outcome of serial dictatorship is ex post efficient. If we also require anonymity, then a maximum utility matching for utilities consistent with the ordinal preferences is also SD-efficient and hence ex post efficient.
\end{remark}

\begin{remark}
An ex post random assignment can have multiple Pareto optimal decompositions. Consider three agents with identical preferences. 
It can even be the case that an ex post efficient lottery can be expressed by a lottery over Pareto dominated assignments ~\citep[Example 2,][]{AbSo03a}. 
\end{remark}


\section{Ex post efficiency}

We consider the problem of testing ex post efficiency of a random assignment. Since we are interested in checking whether a random assignment can be decomposed into Pareto optimal assignment, we are reminded of Birkhoff's algorithm that can decompose any given random assignment (represented by a bistochastic matrix) into a convex combination  of at most $n^2-n+1$ deterministic assignments (represented by permutation matrices)~\citep{LoPl09a}. 

\paragraph{Birkhoff's Algorithm}
Birkhoff's algorithm works as follows. We initialize $i$ to $1$. For a bistochastic matrix $M$, a permutation matrix $P_i$ with respect to $M$ is guaranteed to exist. $M$ is set to $M-\lambda_1{P_i}$ where $\lambda_i\in (0,1)$ is such that no entry $M-\lambda_i{P_i}$ is negative but there is at least one extra zero entry in $M-\lambda{P_i}$ than in $M$. Index $i$ is incremented by one. The updated $M$ is again bistochastic. The process is repeated (say $k-1$ times) until $M$ is the zero matrix. Then $M=\sum_{i=1}^k \lambda_iP_i$. When Birkhoff's algorithm identifies a permutation matrix $P_i$, multiplies it by a constant and then subtracts it from the input matrix representing the random assignment, we will refer to such a step as \emph{decomposing with respect to a permutation matrix}.
\newline

If a random assignment is SD-efficient, then \emph{any} decomposition of the assignment is a decomposition into Pareto optimal deterministic assignments. On the other hand, if a random assignment is not SD-efficient, it may not admit a decomposition into Pareto optimal deterministic assignments.
One may wonder whether we can modify Birkhoff's algorithm to check whether a given random assignment is ex post efficient or not. 
Note that if we decompose with respect to some permutation matrices, then we can decompose in any order.

We first show that checking whether there exists a Pareto optimal permutation matrix that can be used to further decompose an assignment is NP-complete.

\begin{theorem}
Checking whether there exists a Pareto optimal deterministic assignment consistent with the given random assignment is NP-complete.
\end{theorem}
\begin{proof}
 The following problem is NP-complete: {\sc SerialDictatorshipFeasibility} --- check whether there exists a permutation of agents for which serial dictatorship gives a particular object to an agent~\citep{SaSe13a}. We present a reduction from {\sc SerialDictatorshipFeasibility} to the problem of checking whether there exists a Pareto optimal deterministic assignment consistent with the given random assignment
Consider a random assignment $p$ in which agent $i$ gets $o$ with probability one and all other agents gets each object in $O\setminus \{o\}$ with non-zero probability. 
We argue that {\sc SerialDictatorshipFeasibility} has a yes instance if and only if $p$ admits a Pareto optimal deterministic assignment consistent with it.

If there exists a Pareto optimal deterministic assignment consistent with $p$, then this implies that there exists a Pareto optimal assignment in which  $i$ gets $o$. This means that there exists some permutation $\pi$ such that $Prio(N,O,\pref,\pi)(i)(o)=1$. 

Assume that there exists some  permutation $\pi$ such that $Prio(N,O,\pref,\pi)(i)(o)=1$. Then the deterministic assignment corresponding to  $Prio(N,O,\pref,\pi)$ is a  Pareto optimal permutation matrix that is consistent with $p$.
\end{proof}

The statement above does not imply that checking whether a random assignment is ex post efficient is NP-complete. It is even not clear whether testing ex post efficiency of a random assignment is in NP. The reason is that the certificate for membership may not in principle be polynomial-sized. However, we show that verifying an ex post efficient random assignment is in NP but the problem is NP-complete.

	\begin{theorem}\label{mainthm}
	Testing ex post efficiency of a random assignment is NP-complete.
	\end{theorem}

	\begin{proof}
		
		We first show that testing ex post efficiency of a random assignment is in NP. It is sufficient to show that an ex post efficient random assignment admits a polynomial-sized Pareto optimal decomposition. Consider the $n^2$ dimensional Euclidean space of all $n\times n$ matrices. Let $P$ denote the minimal polytope containing all determinisitc Pareto optimal allocations. By definition $P$ is the set of all ex post efficient allocations. Carathéodory's theorem  implies that any ex post efficient allocations must be a convex combination of no more than $n^2+1$ deterministic Pareto optimal allocations (because all vertices in $P$ are Pareto optimal allocations).\footnote{The polytope actually lives in a $n^2-2n+1$ dimensional subspace of $R^{n^2}$ because its feasible points has to satisfy $2n$ equality constraints and one of which is redundant. }

		 We prove the NP-hardness via a reduction from 3-SAT. Given a 3-SAT instance $F = (X, \mc)$, where $X=\{x_1,\ldots, x_k\}$ is the set of binary variables and  $\mc=\{C_1,\ldots,C_t\}$ is the set of clauses. Let each clause $C_j=l_{j,1}\vee l_{j,2}\vee l_{j,3}$, where  for $s=1,2,3$, $l_{j,s}$ is either $x_{j_s}$ or $\neg x_{j_s}$ with $j_1<j_2<j_3$. The assumption that $j_1<j_2<j_3$ will be crucial in the proof.

	Given a 3-SAT instance $F$, we build an assignment problem $(N, O, \succ)$ as follows.

	Let $N_1=\cup_{i=1}^k\{x_i,x_i^1,\ldots,x_i^t\}\cup \{c,c_1,\ldots,c_t\}$. For each agent $x\in N_1$, let $d_x$ denote the corresponding dummy agent. Let $N_2=\{d_x:x\in N_1\}$ and $N=N_1\cup N_2$. That is, $N=\cup_{i=1}^k\{x_i,x_i^1,\ldots,x_i^t, d_{x_i}, d_{x_i^1},\ldots,d_{x_i^t}\}\cup \{c,c_1,\ldots,c_t,d_c,d_{c_1},\ldots,d_{c_t}\}$. We will show that in the decompositions $c_j$'s are ``copies'' of $c$ and $x_i^j$'s are ``copies'' of $x_i$.

	$O=\{+x, -x:\forall x\in N_1\}$. For each $x\in N_1$, $p(x,+x)=p(x,-x)=p(d_x,+x)=p(d_x,-x)=1/2$.

	To define the preferences of the agents, we first introduce the following notation. For any literal $l_{j,s}$ in clause $j$, we let $V(l_{j,s})$ denote the item that corresponds to the value of $x_{j_s}$ that fails $l_{j,s}$. More precisely, 
	$$V(l_{j,s})=\left\{\begin{array}{rl}+{x_i^j}&\text{if }l_{j,s}=\neg x_i\\
	-{x_i^j}&\text{if }l_{j,s}=x_i
	\end{array}\right.$$

	For each $j\leq t$ and $C_j=l_{j,1}\vee l_{j,2}\vee l_{j,3}$, where $l_{j,s}$ is a literal of variable $x_{j_s}$, we let $S_{j_1}^j=\{V(l_{j,2})\}$, $S_{j_2}^j=\{V(l_{j,3})\}$, and $S_{j_3}^j=\{+{c_j}\}$.  For any $i\leq k$ and $j\leq t$, if $S_i^j$ is not defined above then $S_i^j=\emptyset$. Moreover, we let $S=\{V(l_{j,1}):\forall j\leq t\}$.

	For example, for $k=5$, $t=2$, $C_1=x_2\vee \neg  x_4\vee x_5$, and $C_2=\neg  x_2\vee  x_3\vee x_4$, we have 

	$$S_2^1=\{+{x_4^1}\}, S_4^1=\{-{x_5^1}\}, S_5^1=\{+{c_1}\}$$
	$$S_2^2=\{-{x_3^2}\}, S_3^2=\{-{x_4^2}\}, S_4^2=\{+{c_2}\}$$
	$$S=\{-{x_2^1},+{x_2^2}\}$$

	Agents' preferences are defined in two tables: preferences for $x$'s are in Table~\ref{tab:x} and preferences for all $c$'s are in Table~\ref{tab:c}.
	\setlength\extrarowheight{2pt}
	\begin{table}[htp]
	\centering
	\begin{tabular}{rl}
	\toprule $x_i:$&$+x_i,+x_i^1,\ldots,+x_i^t,-x_i,\text{others}$\\
		\midrule $d_{x_i}:$& $+x_i,-x_i,\text{others}$\\
		\midrule $\forall j\leq t$, $x_i^j:$&$\left\{\begin{array}{rl}S_i^j, -x_i^j, -x_i,+x_i^j, \text{others}&\text{ if }x_i\in C_j\\
	-x_i^j, -x_i,S_i^j, +x_i^j, \text{others}&\text{ if }x_i\not\in C_j\end{array}\right.$\\
	\midrule$\forall j\leq t$, $d_{x_i^j}:$&$-x_i^j,+x_i^j,\text{others}$\\
	\bottomrule
	\end{tabular}
	\caption{Preferences for $x$'s. \label{tab:x}}
	\end{table}

	\begin{table}[H]
	\centering
	\begin{tabular}{rl}
	\toprule $c:$&$S,+c,+c_1,\ldots,+c_t,- c,\text{others}$\\
	\midrule $d_c:$& $+c,-c,\text{others}$\\
	\midrule $c_j:$&$-c_j, -c, +c_j, \text{others}$\\
	\midrule $d_{c_j}:$&$-c_j,+c_j,\text{others}$\\
	 \bottomrule
	\end{tabular}
	\caption{Preferences for $c$'s. \label{tab:c}}
	\end{table}

	Because for any $x\in N_1$, $+x$ and $-x$ are assigned to agents $x$ and $d_x$ in $p$ with probability $1$, if $p$ is ex post efficient, then for any deterministic Pareto optimal assignment in the decomposition of $p$, $+x$ and $-x$ must be assigned to $x$ and $d_x$. Therefore, for any such assignment $M$, we say that the {\em sign} of $x$ (respectively, $d_x$) is positive, if  $+x$ is allocated to $x$ (respectively, $d_x$); otherwise the sign is negative. 

	\begin{claim}\label{claim:sign} If $p$ is ex post efficient, then in any deterministic Pareto optimal assignment $M$ in the decomposition of $p$,
	\begin{enumerate}
	\item for all $x\in N_1$, the sign of $x$ is different from the sign of $d_x$;
	\item the sign of $c$ is the same as the sign of $c_j$ for all $j\leq t$;
	\item for all $i\leq k$, the sign of $x_i$ is the same as the sign of $x_i^1,\ldots, x_i^t$.
	\end{enumerate}
	\end{claim}
	\begin{proof}
	Part 1 follows after the fact that in $M$, $+x$ and $-x$ must be assigned to $x$ and $d_x$. 

	For part 2, suppose in $M$ the sign of $c$ is negative and the sign of $c_j$ is positive, then $c$ prefers $+c_j$ and $c_j$ prefers $-c$ (see Table~\ref{tab:c}), which is a trading cycle and contradicts the assumption that $M$ is Pareto optimal. If in $M$ the sign of $c$ is positive and the sign $c_j$ is negative for some $j$, then there exists another deterministic Pareto optimal assignment $M'$ where the sign of $c$ is negative and the sign of $c_j$ is positive. This is because the probability for positive and negative signs for all agents are $1/2$. Then, the same argument can be applied $M'$.

	The proof for part 3 is similar.
	\end{proof}

	In light of Claim~\ref{claim:sign} in the remainder of this proof, we sometimes only use signs to represent the items, which will be clear from the context.

	Suppose $p$ is ex post efficient. We now show that there exists a solution to the 3-SAT instance. Let $M$ be any deterministic Pareto optimal assignment in $p$'s decomposition where the sign of $c$ is positive. In the SAT instance, we let $x_i=+$ if and only if the sign of $x_i$ is positive in $M$ (or equivalently, $x_i$ is assigned item $+x_i$). Suppose for the sake of contradiction a clause $C_j=l_{j,1}\vee l_{j,2}\vee l_{j,3}$ is not satisfied, where $l_{j,s}$ corresponds to variable $x_{j_s}$. By part 2 of Claim~\ref{claim:sign}, $V(l_{j,s})$ is allocated to $x_{j_s}$. Then, in $M$ there exists a trading cycle illustrated in Figure~\ref{fig:clause}, which is a contradiction.  In Figure~\ref{fig:clause} $a_1(o_1)\ra a_2(o_2)$ means that currently $o_1$ (respectively, $o_2$) is allocated to $a_1$ (respectively, $a_2$), and  $a_1$ prefers $o_2$ to $o_1$. Specifically, for $s=1,2,3$, if $l_{j,s}=x_{j_s}$, then $V(l_{j,s})=-x_{j_s}^j$, and $x_{j_s}$ prefers $S_{j_s}^j$ to $-x_{j_s}^j$ (Table~\ref{tab:x}); similarly if $l_{j,s}=\neg x_{j_s}$, then $V(l_{j,s})=+x_{j_s}^j$, and $x_{j_s}$ prefers $S_{j_s}^j$ to $+x_{j_s}^j$ (Table~\ref{tab:x}).

	Therefore, the 3-SAT instance is satisfiable.


	\begin{figure}[htp]
	\centering
	\includegraphics[trim=0 14cm 7.5cm 0, clip=true, width=.7\textwidth]{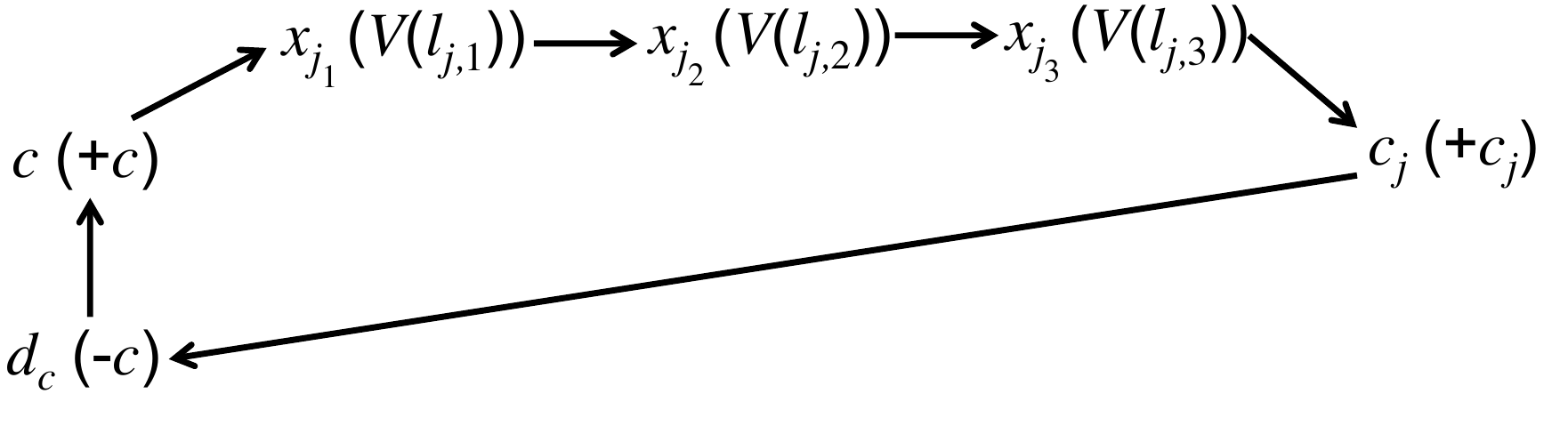}
	\caption{\small A trading cycle through an unsatisfied clause. $a_1(o_1)\ra a_2(o_2)$ means that currently $o_1$ (respectively, $o_2$) is allocated to $a_1$ (respectively, $a_2$), and $a_1$ prefers $o_2$ to $o_1$.\label{fig:clause}}
	\end{figure}

	Suppose there exists a valuation $v$ that satisfies $F$. We now construct a decomposition of $p$ to two deterministic Pareto optimal allocations $M_1$ and $M_2$. $M_1$ is illustrated in Table~\ref{tab:m1}.\begin{table}[htp]
	\centering
	\scalebox{0.85}{
	\begin{tabular}{|c|c|c|c|}
	\hline $c=c_1=\cdots= c_t$&$d_c=d_{c_1}=\cdots =d_{c_t}$& $\forall i$, $x_i=x_i^1=\cdots =x_i^t$ & $\forall i$, $d_{x_i}=d_{x_i^1}=\cdots =d_{x_i^t}$\\
	\hline $+$ & $-$&$v(x_i)$&$\neg v(x_i)$\\ \hline
	\end{tabular}
	}
	\caption{\small $M_1$. \label{tab:m1}}
	\end{table}

	$M_2$ is obtained from $M_1$ by taking the negation of all signs. More precisely, $M_2$ is illustrated in Table~\ref{tab:m2}.
	\begin{table}[htp]
	\centering
	\scalebox{0.85}{
	\begin{tabular}{|c|c|c|c|}
	\hline $c=c_1=\cdots =c_t$&$d_c=d_{c_1}=\cdots= d_{c_t}$& $\forall i$, $x_i=x_i^1=\cdots =x_i^t$ & $\forall i$, $d_{x_i}=d_{x_i^1}=\cdots =d_{x_i^t}$\\
	\hline $-$ & $+$&$\neg v(x_i)$&$v(x_i)$\\ \hline
	\end{tabular}
	}
	\caption{\small $M_2$. \label{tab:m2}}
	\end{table}
	It is easy to check that $p=\frac{1}{2}M_1+\frac{1}{2}M_2$. The demand graph of $x$'s is illustrated in Figure~\ref{fig:demandx}, where all outgoing edges of $x$'s are shown (some incoming edges are not shown). We recall that an edge from agent $a_1$ to agent $a_2$ means that $a_1$ demand the item allocated to $a_2$. (a) represents the case for $x_i=+$ and (b) represents the case for $x_i=-$. Dashed lines in (b) means that it is valid if and only if $x_i$ is a literal in $C_j$.
	\begin{figure}[htp]
	\centering
	\begin{tabular}{cc}
	\includegraphics[trim=0 3cm 13cm 8cm, clip=true, width=.5\textwidth]{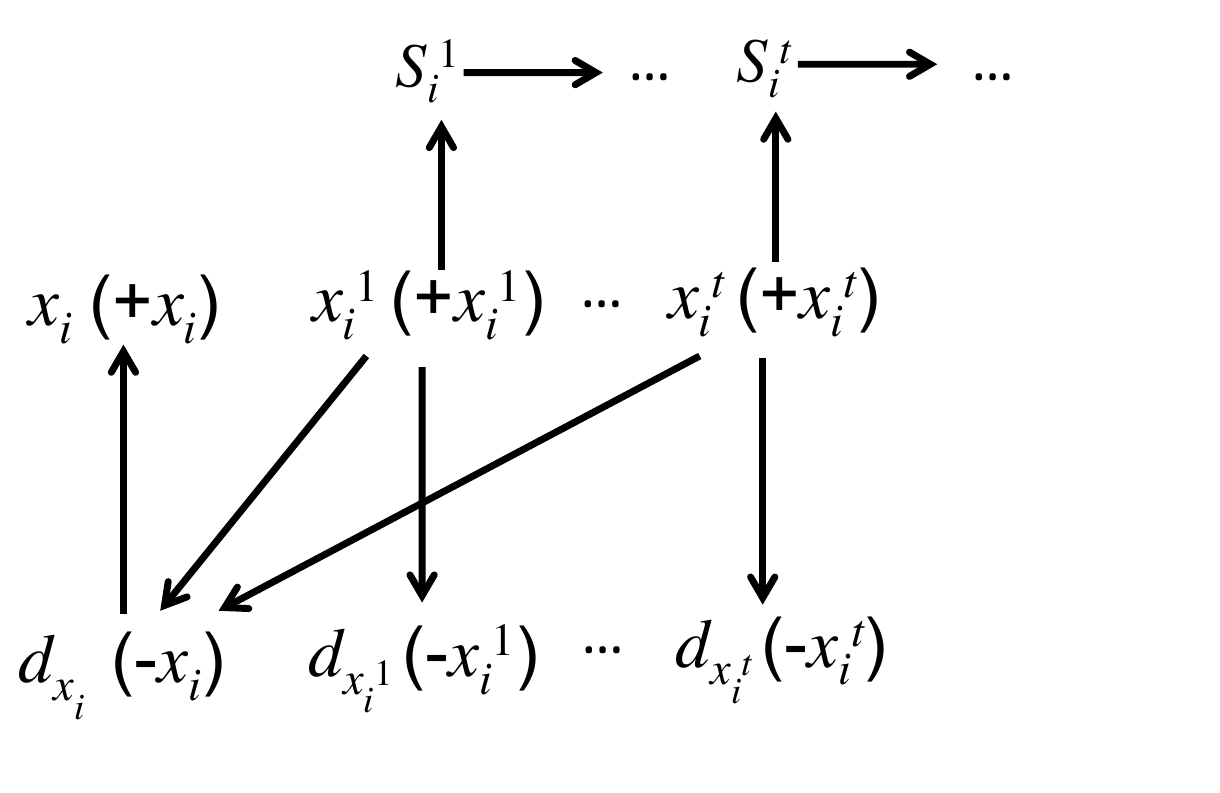}& \includegraphics[trim=13cm 3cm 0cm 8cm, clip=true, width=.5\textwidth]{Demand.pdf} \\
	(a)& (b)
	\end{tabular}
	\caption{\small The demand graph of $x_i$'s in $M_1$ and $M_2$. All outgoing edges of $x's$ and dummies are shown. In (b), $S_i^j\ra x_i^j$ if and only if $x_i\in C_j$.\label{fig:demandx}}
	\end{figure}

	\begin{claim} For all $\leq k$ and $j\leq t$, $x_i$, $d_{x_i}$ and $d_{x_i^j}$ are not involved in any trading cycle.
	\end{claim}
	\begin{proof}
	For Figure~\ref{fig:demandx} (a), no cycle can involve $x_i$, $d_{x_i^1},\ldots,d_{x_i^t}$ because these agents have their top items. Then, $d_{x_i}$ cannot be in any cycle because its only outgoing edge is to $x_i$, which is not in any cycle.

	For Figure~\ref{fig:demandx} (b), no cycle can involve $x_i$ because the only agents who may demand $-x_i$ are $x_i^j$'s with $+x_i^j$ (Table~\ref{tab:x} and \ref{tab:c}), but $x_i^j$'s get $-x_i^j$'s in Figure~\ref{fig:demandx} (b). Also no cycle can involve $d_{x_i}$ because she has her top item. For any $j$, if $d_{x_i^j}$ is involved in a cycle, then there is exactly one agent beyond $x_i$ who demands $+x_i^j$, who is the preceding agent in $C_j$ (which can be another $x_{i'}^j$ or $c$). However, if an agent demands $+x_i^j$, then $\neg x_i\in C_j$, which implies $x_i\not\in C_j$. Hence there is no edge from $x_i^j$ to $S_i^j$ in Figure~\ref{fig:demandx} (b) (see Table~\ref{tab:x}). In this case $x_i^j$ gets her top item, and because the only outgoing edge of $d_i^j$ is from $x_i^j$, it is impossible for $d_{x_i^j}$ to be involved in a trading cycle, which is a contradiction.
	\end{proof}

	We establish that both $M_1$ and $M_2$ are Pareto optimal in the following two claims.

	\begin{claim}\label{claim:m1} $M_1$ is Pareto optimal.
	\end{claim}
	\begin{proof} The demand graph of $c's$ in $M_1$ with all outgoing edges of $c$'s is illustrated in Figure~\ref{fig:demandc} (a).
	\begin{figure}[htp]
	\centering

	\end{figure}

	\begin{figure}[htp]
	\centering
	\begin{tabular}{cc}
	\includegraphics[trim=0 12cm 16cm 0, clip=true, width=.4\textwidth]{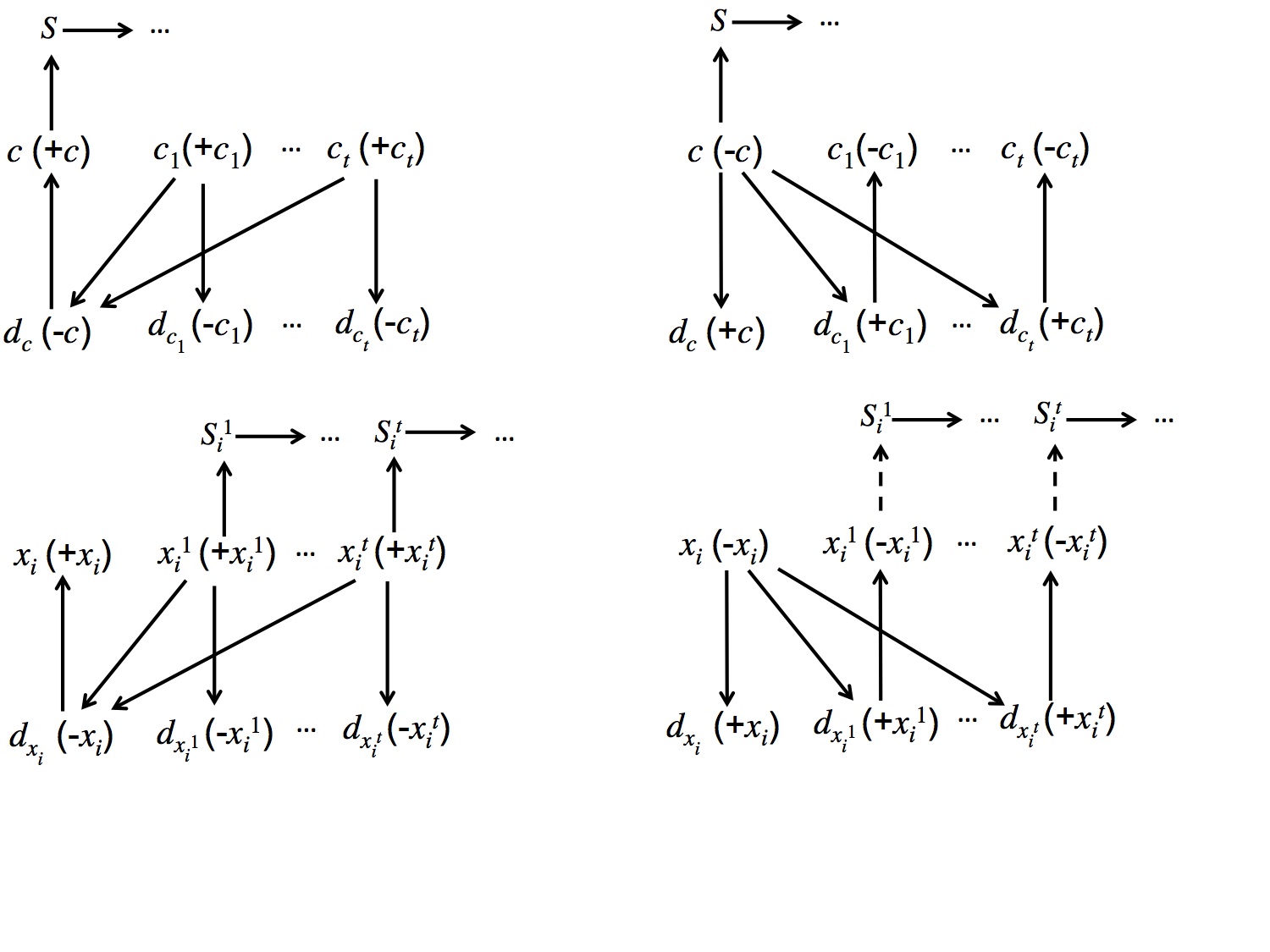}& \includegraphics[trim=11cm 12cm 2cm 0cm, clip=true, width=.5\textwidth]{Demand.jpg} \\
	(a) $M_1$. & (b) $M_2$. 
	\end{tabular}
	\caption{\small The demand graph of $c$'s in $M_1$ and $M_2$. All outgoing edges of $c's$ and dummies are shown.\label{fig:demandc}}
	\end{figure}

	Clearly for all $j\leq t$, $d_{c_j}$ is not in any cycle because they have no outgoing edges. The only possibility of cycles are through $c\ra d_c\ra c_j$, then through truncated Figure~\ref{fig:demandx}, where all $x_i$ and $d_{x_i^j}$ are removed. Because $v$ satisfies $F$, each potential path from $c_j$ to $c$ is blocked by at least one $x_i^j$. Meanwhile, there is no cycle involving only $x_i^j$'s or involving $x_i^j$'s with different $j$'s,  because any $x_i^j$ only has outgoing edges to $S_i^j$, which is either (1) empty,  or (2) contains another $x_{i'}^j$ with $i'>i$ (recall we assume that $j_1<j_2<j_3$ for the three literals in $C_j$) or $+c_j$. This proves that $M_1$ is Pareto optimal.
	\end{proof}

	\begin{claim} $M_2$ is Pareto optimal.
	\end{claim}
	\begin{proof} The demand graph of $c's$ in $M_2$ with all outgoing edges of $c$'s is illustrated in Figure~\ref{fig:demandc} (b). Clearly $d_c$ and all $c_j$'s are not in any cycle because they have no outgoing edges. After they are removed, $d_{c_j}$'s have no outgoing edges. $c$ can also be removed because none of the remaining agents demand $-c$. The only remaining agents are those in the truncated Figure~\ref{fig:demandx}, where all $x_i$ and $d_{x_i^j}$ are removed, and as in Claim~\ref{claim:m1}, there is no cycle among them.
	\end{proof}
	\end{proof}
	
	\begin{corollary}
		Checking membership of a point in $\mathcal{P}$ is NP-complete. 
		\end{corollary}
	\begin{proof}
		Since $\mathcal{P}$ is a convex combination of Pareto optimal determistic assignments, it contains all the ex post efficient points.
		\end{proof}
		
		\begin{corollary}
Optimizing a linear functions over $\mathcal{P}$ is NP-complete.
			\end{corollary}
		\begin{proof}
			Since testing membership in $\mathcal{P}$, the statements follows from the equivalence between optimizing over a poltope and implementing a separation oracle over over a polytope~\citep{GLS93a}.
			\end{proof}
	
	Although we have shown that testing whether a given assignment is ex post efficiency is NP-complete, we leave open the case when the assignment is uniform. 

We now show that in the random assignment problem, ex post efficiency is not combinatorial.\footnote{The notion of an efficiency concept being combinatorial was first discussed in \citep{ABB14a}. However, the setting was voting and not the random assignment problem. In voting, a lottery over alternatives is ex post efficient iff the support consists of Pareto optimal alternatives. Hence in voting, ex post efficiency is combinatorial.} 
			 This is already a contrast with a stronger notion of efficiency called \emph{SD-efficiency} which depends solely on the support of the random allocations. 
A trading cycle of size $2k$ is \emph{consistent} with random assignment $p$ if it consists of $k$ agents and $k$ objects, each
objects points to an agent, each object to an agent in the cycle and the cycle satisfies the following constraints: $(a)$ $p(i)(o)>0$ if 
object $o$ points to the agent $i$, and $(b)$  $o' \succ_i o$ if agent $i$ points to object $o'$.
\citet[Lemma 3, ][]{BoMo01a} proved that a random assignment is not SD-efficiency iff it admits a trading cycle consistent with it. The fact that SD-efficiency is combinatorial follows from its characterization. In contrast, ex post efficiency is not combinatorial.

			\begin{theorem}\label{th:expost-not-combinatorial}
		Ex post efficiency is not combinatorial.
			\end{theorem}
		\begin{proof}
		Consider the following assignment problem.
		\begin{align*}
		1:&\quad o_1,o_2,o_3,o_4\\
		2:&\quad o_1,o_2,o_3,o_4\\
		3:&\quad o_2,o_1,o_4,o_3\\
		4:&\quad o_2,o_1,o_4,o_3
		\end{align*}

		The following assignment $p$ is a result of RSD and hence ex post efficient:

		\[p=\begin{pmatrix}
			\nicefrac{5}{12}&\nicefrac{1}{12}&\nicefrac{5}{12}&\nicefrac{1}{12}\\
		\nicefrac{5}{12}&\nicefrac{1}{12}&\nicefrac{5}{12}&\nicefrac{1}{12}\\
		\nicefrac{1}{12}&\nicefrac{5}{12}&\nicefrac{1}{12}&\nicefrac{5}{12}\\
		\nicefrac{1}{12}&\nicefrac{5}{12}&\nicefrac{1}{12}&\nicefrac{5}{12}
			\end{pmatrix}.\]
			\
			Now consider the following assignment:


			\[q=\begin{pmatrix}
				\nicefrac{1}{12}&\nicefrac{5}{12}&\nicefrac{1}{12}&\nicefrac{5}{12}\\
			\nicefrac{1}{12}&\nicefrac{5}{12}&\nicefrac{1}{12}&\nicefrac{5}{12}\\
			\nicefrac{5}{12}&\nicefrac{1}{12}&\nicefrac{5}{12}&\nicefrac{1}{12}\\
			\nicefrac{5}{12}&\nicefrac{1}{12}&\nicefrac{5}{12}&\nicefrac{1}{12}
				\end{pmatrix}.\]

			Note that $q$ is a random assignment such that $q(i)(o)>0$ if and only if $p(i)(o)>0$.
			Then in each Pareto optimal matrix that can be used to decompose the assignment, due to Pareto optimality, agents $1$ and $2$ get an object each from the following sets of objects
		$\{o_1,o_2\}$, $\{o_1,o_3\}$, and $\{o_3,o_4\}$.
		Agent 1 and 2 cannot get the set $\{o_2,o_4\}$ in a Pareto optimal assignment because if it were the case then there exists a serial dictatorship in which agent 3 or 4 take $o_1$ before $o_2$ is allocated. In order for agent $1$ to get $5/12$ of $o_4$, we need to use the $5/12$ time the Pareto optimal assignment in which $o_4$ gts $5/12$. But this means that agent $2$ gets $o_3$ $5/12$ of the time. But this is not possible since agent $2$ gets $o_3$ $1/12$ of time. 
%
		\end{proof}

\section{Robust ex post efficiency}

In this section, we formalize a new efficiency concept called \emph{ex post efficiency}. Recall that a random assignment is ex post efficient if it can be represented as a convex combination of Pareto optimal deterministic assignments. We say that a random assignment is \emph{robust ex post efficient} if any decomposition of the assignment consists of Pareto optimal deterministic assignments. The following is a useful characterization of robust ex post efficiency.

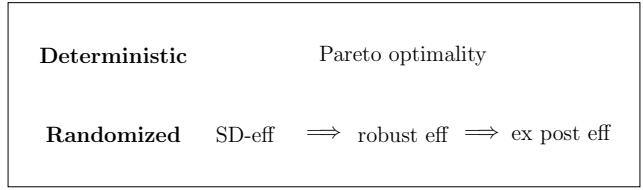
\begin{figure}[htb]

\centering
\scalebox{0.70}{
\begin{tikzpicture}
\tikzstyle{pfeil}=[->,>=angle 60, shorten >=1pt,draw]
\tikzstyle{onlytext}=[]

	\draw (-6,-1) rectangle (6,2.5);


\node[onlytext] (po) at (-4,1.5) {\large \textbf{Deterministic}};
\node[onlytext] (po) at (1.5,1.5) {\large Pareto optimality};

\node[onlytext] (po) at (-4,0) {\large\textbf{Randomized}};
\node[onlytext] (sdeff) at (-1.5,0) {\large SD-eff};
\node[onlytext] (sdeff) at (0,0) {\large$\implies$};
\node[onlytext] (robusteff) at (1.5,0) {\large robust eff};
\node[onlytext] (imp) at (3,0) {\large $\implies$};
\node[onlytext] (expost) at (4.5,0) {\large ex post eff };


\end{tikzpicture}
}
\caption{Efficiency concepts for deterministic and randomized assignment settings.}
\label{fig:relations}

\end{figure}

\begin{theorem}\label{th:robust-eff-charac}
	An assignment is robust ex post efficient iff it does not admit a non-Pareto optimal deterministic assignment  consistent with it.
\end{theorem}
\begin{proof}
	Assume that there exists no non-Pareto optimal deterministic assignment  consistent with the assignment $p$. Then each time, $p$ is decomposed with respect to a deterministic assignment, it is with respect to a Pareto optimal deterministic assignment. By Birkhoff's theorem, the updated $p$ is still bistochastic after the decomposition. Hence any decomposition of $p$ is into Pareto optimal deterministic assignment and hence $p$ is robust ex post efficient.

Now let us assume that that there exists a non-Pareto optimal deterministic assignment  consistent with the assignment $p$. We then decompose $p$ with respect to such an assignment using Birkhoff's algorithm and continue decomposing it. Since the decomposition contains at least non Pareto optimal deterministic assignment, hence $p$ is not robust ex post efficient.
\end{proof}

We point out that SD-efficiency implies robust ex post efficiency which implies ex post efficiency. Note that in the general domain of voting, ex post efficiency and robust ex post efficiency are equivalent.

\begin{theorem}
SD-efficiency implies robust ex post efficiency which implies ex post efficiency. 	
\end{theorem}
\begin{proof}
	We first show that if an assignment is not ex post efficient, then it is not robust ex post efficient.
If an assignment $q$ is not ex post efficient, then there does not exist \emph{any} decomposition of the assignment into Pareto optimal deterministic assignments. Then $q$ is not robust ex post efficient.

	We now show that if an assignment is not robust ex post efficient, then it is not SD efficient.
If an assignment $q$ is not robust ex post efficient, then there exists at least one decomposition of $q$ into deterministic assignments in which at least one deterministic assignment is not Pareto optimal. But this implies that there exists at least one deterministic assignment $M$ consistent with $q$ that is not Pareto optimal. By Fact~\ref{fact:po:ttc}, $M$ admits a trading cycle $C$. Since $M$ is consistent with $q$, each object $o$ that points to an agent $i$ in the cycle is such that $q(i)(o)>0$.
Now consider a random assignment $p$ which is the same as $q$ except that each agent $i$ gets $\epsilon$ probability less of the object that was pointing to it and $\epsilon$ probability more of the object he points to in cycle $C$. Assignment $q$ is such that
$p(i) \succsim_i^{SD} q(i)$ for all $i\in N$ and $p(i) \succ_i^{SD} q(i)$ for all $i\in C$. Hence $q$ is not SD-efficient.
\end{proof}

	SD-efficiency is a strictly stronger concept than robust ex post efficiency. Although \citet{AbSo03a} did not explicitly define the concept robust ex post efficiency, they showed that a random assignment that has only one decomposition which is a randomization over Pareto optimal assignments does not satisfy SD-efficiency. Hence robust ex post efficiency does not imply SD-efficiency.
Next, we show that ex post efficiency does not imply robust ex post efficiency so that robust  ex post efficiency is a strictly stronger concept than ex post efficiency.

\begin{remark}\label{remark:expostnotrobust}
	Consider the assignment problem 
	assignment $p$ in the proof of Theorem~\ref{th:expost-not-combinatorial}. Assignment $p$ which is the result of random serial dictatorship
is ex post efficient. However, there is a deterministic assignment $M$ consistent with $p$ that is not Pareto optimal. Hence $p$ is not robust ex post efficient.
\end{remark}

Hence, we get the following.

\begin{theorem}
	The random serial dictatorship mechanism may return an assignment that is not robust ex post efficient.
\end{theorem}

The simple theorem above strengthens the observation of \citet{BoMo01a} that random serial dictatorship is not SD-efficient.
%

\begin{theorem}
Any robust ex post efficient point must lie on a face of the assignment polytope where all extreme points of the face are vectors of Pareto optimal assignments. 
\end{theorem}
\begin{proof}
Suppose not, then there is a robust ex post efficient point $x$ that is in the interior of the assignment polytope.  Then there is a $\epsilon$ ball centered around $x$ that is contained in the polytope. Let $p$ be an extreme point that is not Pareto optimal. Choose $\delta$ small enough such that 
$y = x + \delta(x-p)$ lies in the $\epsilon$ ball. Then one can write $x$ as a convex combination of $y$ and $p$. Moreover, since $y$ belongs to the assignment polytope, $y$ can in turn be written as a convex combination of vectors of assignments. Hence, we have expressed x as a convex combination of extreme points of the assignment polytope and one of which is not a Pareto optimal assignment. So $x$ cannot be robust ex post efficient. Hence, any robust ex post efficient must lie on a face of the assignment polytope. We can then use the same argument to say that all extreme points of the face must correspond to Pareto optimal assignments.
\end{proof}

We can consider similar computational questions regarding robust ex post efficiency: \emph{what is the computational complexity of checking whether an assignment is robust ex post efficient?} Our first observation is that the problem is in coNP.

\begin{remark}\label{remark:incoNP}
	The problem of checking whether a random assignment is robust ex post efficient is in coNP. By Theorem~\ref{th:robust-eff-charac},	any non-Pareto optimal deterministic assignment consistent with the random assignment is a witness that the random assignment is not robust ex post efficient. Also note that it can be checked in linear time whether a given assignment is Pareto optimal (Fact~\ref{fact:po:ttc}). 
\end{remark}
%

Due to the characterization of robust efficiency in Theorem~\ref{th:robust-eff-charac}, the problem of testing robust ex post efficiency is equivalent to  
checking whether there exists a constrained non Pareto optimal assignment. Previously, it has been shown that checking whether there exists a constrained Pareto optimal assignment is NP-complete~\citep{SaSe13a}. Next we give a simple necessary condition for robust ex post efficiency.



\begin{remark}
	If a random assignment is robust ex post efficient, there exists no consistent deterministic assignment in which no agent gets his most preferred object. 
	The argument is as follows.
	For a random assignment $p$, it is sufficient to show that if there exists a consistent deterministic assignment $q$ in which no agent gets his most preferred object, then that assignment is not Pareto optimal and hence $p$ is not robust ex post efficient. If no agent gets their most preferred object, then there exists no permutation over the agents under serial dictatorship returns $p$. Therefore, by Fact~\ref{fact:po:sd}, $q$ is not Pareto optimal. Hence by Theorem~\ref{th:robust-eff-charac}, $p$ is not robust ex post efficient.
\end{remark}

Next, we show that robust ex post efficiency is combinatorial. 

		\begin{theorem}\label{th:robust-is-combinatorial}
	Robust ex post efficiency is combinatorial.
		\end{theorem}
		\begin{proof}
If $p$ is not robust ex post efficient, then there exists a deterministic assignment $r$ where $p(i)(o)>0$ for	$r(i)(o)=1$ and $r$ is not Pareto optimal. Hence for deterministic assignment $r$,   $q(i)(o)>0$ for	$r(i)(o)=1$.
Thus $r$ is not robust ex post efficient because it admits a Pareto dominated deterministic assignment that is consistent with it. The same argument also shows that is $q$ is not robust ex post efficient, then $p$ is not robust ex post efficient.
		\end{proof}

		In the previous section, we mentioned that the  complexity of testing whether the uniform assignment is ex post efficient is still open. On the other hand, it can be easily checked whether the uniform assignment is robust ex post efficient.

		\begin{remark}
		The uniform assignment is robust ex post efficient if and only if the preferences are unanimous. The arguments is as follows. 	If preferences are unanimous then every assignment is Pareto optimal. 
			If preferences are not unanimous, then for some two objects $o$ and $o'$, at least two agents have opposite preferences over them. In this case, any assignment that gives each of $o$ and $o'$ to one of the two agents who prefers it less is not Pareto optimal.
		\end{remark}
		
		Two agents are said to of the same \emph{type} if they have identical preferences. 
 We show that if there are a constant number agents, types, then robust ex post efficiency can be checked in polynomial time.

		\begin{lemma}\label{lemma:agent-types}
			If there is a trading cycle consistent with random assignment that contain multiple agents of the same type, then there also exists a Pareto cycle consistent with the random assignment, in which there is at most one agent of the same type.
		\end{lemma}
		\begin{proof}
			We show that if there exists a Pareto cycle also containing $c$ agents of the same type, then there exists a Pareto cycle containing $c-1$ agents of the same type.
			Consider the Pareto cycle in which the $c$ agents of the same type are $i_1,\ldots, i_c$ where agent $i_1$ is the agent who has the least preferred object pointing to it among all the objects that point towards the $c$ agents. Then agent $i_1$ can point directly to the object that agent $i_2$ points to thereby forming a smaller cycle only including $i_1,i_3,\ldots, i_c$.
		\end{proof}

		\begin{theorem}
			If there are a constant number of agent types, robust ex post efficiency can be checked in polynomial time. 
		\end{theorem}
		\begin{proof}
			We will use the characterization of robust efficiency in Theorem~\ref{th:robust-eff-charac} to propose an algorithm. Since we will check for a Pareto dominated assignment, by Fact~\ref{fact:po:ttc}, such an assignment admits a trading cycle consistent with the random assignment.
From Lemma~\ref{lemma:agent-types}, we know that if there is a trading cycle consistent with the random assignment then there is also a trading cycle consistent with the assignment in which there is at most one agent of one type. In order to check for robust efficiency, we need to check whether there is a trading cycle for which the agents outside the trading cycle are assigned an object each from outside the trading cycle consistent with the random assignment. We will use the fact that if there exists a trading cycle consistent with the assignment which contains multiple agents of the same type and if the agents outside the trading cycle are also perfectly matched to an object consistent with the random assignment, then there exists a trading cycle consistent with the assignment which contains at most one agent of each type and for which the agents outside the trading cycle are also perfectly matched to an object consistent with the random assignment.

Note that are constant number $\sum_{j=1}^{k}{k\choose j}\times j!$ of possible orderings in which agent types are present in a Pareto cycle in which there is at most one agent of the same type. For each of the agents types in the cycle, there may be $O(n)$ options of objects that point to him. For each agent type there are $O(n)$ agents that could be used for that type. Therefore, the number of Pareto cycles to be considered is at most $O(n^{2k})$ where in each Pareto cycle there are at most $k$ agents and $k$ objects.  
					
For each of the (short) Pareto cycles considered that are not more than $O(n^{2k})$, we need to check whether the other agents can be perfectly matched to the unallocated objects. This can again be checked in polynomial time via the algorithm to check whether a perfect matching exists for agents \emph{not} in the Pareto cycle. 
		\end{proof}
		
It will be interesting to check whether there exists a fixed parametrized algorithm with parameter number of agent types.

\section{Conclusions}

%
%
%
%


We examined different aspects of ex post efficiency of random assignments. 
One of the most important technical result in the paper is that testing ex post efficiency is NP-complete. 
The result contrast  with the followings facts
(1) ex post stability in the two sided marriage setting can be tested in polynomial time via an LP~\citep{TeSe98a};
(2) SD-efficiency can be tested in polynomial time; and
(3) an ex post efficient assignment can be computed in polynomial time.

 One implication of the NP-completeness result is that the set of ex post efficient assignments cannot be characterized compactly. Unless P=NP, there is no polynomial-time separation oracle for the convex hull of Pareto optimal assignments. Due to the well-known equivalence of optimization and separation, it follows that optimizing a linear function over the convex hull of Pareto optimal assignments is NP-complete as well.\footnote{The fact that optimizing a linear function over the convex hull of Pareto optimal assignments is NP-complete also follows from \citep{SaSe13a}.}
 Another corollary is that for the marriage market model, testing ex post effiency is NP-complete because the assignment setting can be viewed as a marriage market in which one side of the market is completely indifferent.
 
We also showed that in the random assignment problem, robust ex post efficiency is combinatorial whereas ex post efficiency is not.
The finding that ex post efficiency is not combinatorial contrasts with the fact that in randomized voting, ex post efficiency of a lottery simply depends on its support.

%

A number of open problems arise as a result of this study.
The complexity of checking whether a random assignment is robust ex post efficient is open. Similarly, for a constant number of agent types, the complexity of checking whether a random assignment is ex post efficient is also open. 
Computational aspects of Pareto optimal deterministic assignments have been studied in great depth in recent years~\citep{ACMM05a,ABH11c,Manl13a}. 
The more general randomized resource allocation settings provide a suitable ground for further developments in the algorithmic aspects of matching under preferences.




\section*{Acknowledgments}
NICTA is funded by the Australian Government through the Department of Communications and the Australian Research Council through the ICT Centre of Excellence Program.


\end{document}